\documentclass[11pt]{article}

\usepackage[margin=1in]{geometry}
\usepackage{CJK, hyperref, amsmath, amsthm, amssymb, authblk, enumitem}

\newtheorem{theorem}{Theorem}
\newtheorem{lemma}{Lemma}
\newtheorem{corollary}{Corollary}

\theoremstyle{definition}
\newtheorem{definition}{Definition}

\theoremstyle{remark}
\newtheorem*{remark}{Remark}

\DeclareMathOperator{\poly}{poly}

\DeclareMathOperator{\T}{\mathbb T}

\setlist{nosep}

\begin{document}

\title{Computing local properties in the trivial phase}

\begin{CJK*}{UTF8}{}

\CJKfamily{gbsn}

\author{Yichen Huang (黄溢辰)\thanks{yichuang@mit.edu}}
\affil{Center for Theoretical Physics, Massachusetts Institute of Technology, Cambridge, Massachusetts 02139, USA}

\maketitle

\end{CJK*}

\begin{abstract}

A translation-invariant gapped local Hamiltonian is in the trivial phase if it can be connected to a completely decoupled Hamiltonian with a smooth path of translation-invariant gapped local Hamiltonians. For the ground state of such a Hamiltonian, we show that the expectation value of a local observable can be computed in time $\poly(1/\delta)$ in one spatial dimension and $e^{\poly\log(1/\delta)}$ in two and higher dimensions, where $\delta$ is the desired (additive) accuracy. The algorithm applies to systems of finite size and in the thermodynamic limit. It only assumes the existence but not any knowledge of the path.

\end{abstract}

\section{Introduction and background}

Computing the ground-state expectation value of a local observable in quantum many-body systems is a fundamental problem in condensed matter physics. We consider this problem for translation-invariant gapped local Hamiltonians in systems of finite size and in the thermodynamic limit. Here, ``gapped'' means that the energy gap (defined as the energy difference between the \emph{unique} ground state and the first excited eigenstate) is lower bounded by a positive constant independent of the system size. For notational simplicity, we first work with a chain of $N$ spins, each of which has constant local dimension. In Section \ref{sec:h}, we extend the results to two and higher spatial dimensions.

It is straightforward to define translation-invariant local Hamiltonians in the thermodynamic limit $N\to+\infty$. Without loss of generality, we only consider nearest-neighbor interactions. Let $h_1$ with $\|h_1\|\le1$ be a Hermitian operator acting on the first two spins. Let $\T$ be the (unitary) lattice translation operator so that $h_j:=\T^{j-1}h_1\T^{-(j-1)}$ acts on the $j$th, $(j+1)$th spins. For concreteness, we use open boundary conditions (we will discuss periodic boundary conditions later) and define
\begin{equation} \label{eq:TIH}
H=\left\{H^{(N)}\right\},\quad H^{(N)}:=\sum_{j=1}^{N-1}h_j
\end{equation}
as a sequence of translation-invariant local Hamiltonians: one for each system size $N$.

Do the ground-state properties of $H^{(N)}$ depend smoothly on the system size $N$ and become well defined in the thermodynamic limit $N\to+\infty$? Not always. For example, the ground-state energy (as a function of $N$) can encode the solution of a computationally intractable problem \cite{GI09, GI13, BCO17}. The problem of whether a constant energy gap persists in the thermodynamic limit is undecidable \cite{CPW15, BCLP18}. In size-driven quantum phase transitions, the ground-state properties change abruptly as $N$ crosses the critical system size $N_c$ \cite{BCL+19}. There is good evidence that even under the assumption of a constant energy gap, the ground-state expectation value of a local observable may not always be computed in time independently of $N$ \cite{DB19, Hua19}.

In this paper, we consider the case that $H$ is in the trivial phase.
\begin{definition} [trivial phase; see, e.g., Refs. \cite{CGW10, SPC11, Has12, HC15}] \label{def:tri}
$H=\{H^{(N)}\}$ is in the trivial phase if there exists a smooth path $H(s)=\{H^{(N)}(s)\}$ of translation-invariant local Hamiltonians
\begin{equation}
H^{(N)}(s)=\sum_{j=1}^{N-1}h_j(s),\quad h_j(s):=\T^{j-1}h_1(s)\T^{-(j-1)},\quad0\le s\le1
\end{equation}
such that
\begin{itemize}
\item $H^{(N)}(0)$ is a completely decoupled Hamiltonian so that its ground state is a product state.
\item $H^{(N)}(1)=H^{(N)}$.
\item $\|h_1(s)\|\le1$ and $\|\mathrm dh_1(s)/\mathrm ds\|\le1$.
\item The energy gap of $H^{(N)}(s)$ is lower bounded by a positive constant $\epsilon$ for any $0\le s\le1$ and $N\ge N_0$, where $N_0$ is a constant.
\end{itemize}
\end{definition}

It is widely believed (and argued using uniform matrix product states \cite{FNW92, PVWC07}) that the trivial phase is the only gapped phase in one-dimensional translation-invariant systems \cite{CGW11, SPC11}. Therefore, being in the trivial phase is a very mild assumption in one spatial dimension.

\section{One dimension}

Throughout this paper, asymptotic notations are used extensively. Let $f,g:\mathbb R^+\to\mathbb R^+$ be two functions. One writes $f(x)=O(g(x))$ if and only if there exist positive numbers $M,x_0$ such that $f(x)\le Mg(x)$ for all $x>x_0$; $f(x)=\Omega(g(x))$ if and only if there exist positive numbers $M,x_0$ such that $f(x)\ge Mg(x)$ for all $x>x_0$. To simplify the notation, we use a tilde to hide a polylogarithmic factor, e.g., $\tilde O(f(x)):=O(f(x)\poly|\log f(x)|)$.

Let $A_1$ with $\|A_1\|\le1$ be a local operator acting on the first few spins, and $A_j:=\T^{j-1}A_1\T^{-(j-1)}$ be the lattice-translated copy of $A_1$. Let $\langle\hat O\rangle_N:=\langle\psi^{(N)}|\hat O|\psi^{(N)}\rangle$ be the expectation value of an operator in the ground state $\psi^{(N)}$ of $H^{(N)}$.

\begin{lemma} [open boundary conditions] \label{thm:1}
Suppose that $H=\{H^{(N)}\}$ is in the trivial phase. Then,
\begin{align}
&|\langle A_j\rangle_N-\langle A_j\rangle_{N+1}|=e^{-\tilde\Omega(N-j)},\label{eq:l11}\\
&|\langle A_j\rangle_N-\langle A_{j+1}\rangle_N|=e^{-\tilde\Omega(\min\{j,N-j\})}.\label{eq:l12}
\end{align}
Therefore, both $\lim_{N\to+\infty}\langle A_j\rangle_N$ and $\lim_{N\to+\infty}\langle A_{[\alpha N]}\rangle_N$ are well defined, where $0<\alpha<1$ is a constant and $[\cdots]$ denotes the floor function. Furthermore, the value of the latter limit is independent of $\alpha$.
\end{lemma}

\begin{proof}
We use the technique of quasi-adiabatic continuation, which was originally due to Hastings \cite{Has04, HW05} and subsequently developed in Refs. \cite{Osb07, Has10, BMNS12}. Combining with the Lieb-Robinson bound \cite{LR72, NS06, HK06, WH19}, this technique has applications in proving, e.g., stability of topological order \cite{BHM10, BH11, MZ13} and quantization of the Hall conductance \cite{HM15}.

Following Ref. \cite{Has12}, we give a high-level overview of quasi-adiabatic continuation. Define
\begin{equation} \label{qac}
D^{(N)}(s)=\sum_{j=1}^{N-1}D^{(N)}_j(s),\quad D^{(N)}_j(s)=i\int_{-\infty}^{+\infty} f(t)e^{iH^{(N)}(s)t}\frac{\mathrm dh_j(s)}{\mathrm ds}e^{-iH^{(N)}(s)t}\,\mathrm dt,
\end{equation}
where the ``filter function'' $f(t)$ is purely imaginary so that $D^{(N)}(s)$ is Hermitian. Let $\psi^{(N)}(s)$ be the ground state of $H^{(N)}(s)$. Reference \cite{Has10} explicitly constructed $f(t)$ such that
\begin{align}
&\frac{\mathrm d\psi^{(N)}(s)}{\mathrm ds}=-iD^{(N)}(s)\psi^{(N)}(s),\label{sch}\\
&|f(t)|=e^{-\Omega(t/\log^{1.001}t)}=e^{-\tilde\Omega(t)}.\label{filter}
\end{align}
The first equation allows us to interpret $D^{(N)}(s)$ as a time-dependent Hamiltonian, whose dynamics generates the state $\psi^{(N)}(s)$ from $\psi^{(N)}(0)$. Let $B(j,r)$ be the radius-$r$ neighborhood of the $j$th spin. Choosing an appropriate constant $c$, we split the integral (\ref{qac}) into two terms:
\begin{equation}
D_j^{(N)}(s)=i\int_{|t|\le cr}\cdots+i\int_{|t|>cr}\cdots.
\end{equation}
The first term is approximately supported on $B(j,r)$ due to the Lieb-Robinson bound for $H^{(N)}(s)$, and second term is negligible due to the fast decay (\ref{filter}) of $f(t)$. Hence, $D^{(N)}_j(s)$ is local in the sense that
\begin{equation} \label{eq:local}
\left\|D^{(N)}_j(s)-D^{(N)}_j(s)|_{B(j,r)}\right\|=e^{-\tilde\Omega(r)},\quad\forall j,s,
\end{equation}
where $D^{(N)}_j(s)|_{B(j,r)}$ is the best approximation of $D^{(N)}_j(s)$ supported on $B(j,r)$. The locality (\ref{eq:local}) of $D^{(N)}_j(s)$ implies that the dynamics generated by $D^{(N)}(s)$ also satisfies a Lieb-Robinson bound \cite{Has10}. A similar argument shows that
\begin{equation} \label{eq:error1}
\left\|D^{(N)}_j(s)-D^{(N+1)}_j(s)\right\|=e^{-\tilde\Omega(N-j)}.
\end{equation}
Let
\begin{equation}
U^{(N)}(s):=\mathcal S'e^{-i\int_0^sD^{(N)}(s')\,\mathrm ds'}=\mathcal S'e^{-i\int_0^s\sum_{k=1}^{N-1}D_k^{(N)}(s')\,\mathrm ds'},
\end{equation}
where $\mathcal S'$ is the $s'$-ordering operator. The thermodynamic limit of quasi-adiabatic continuation exists \cite{BMNS12} in the sense that
\begin{equation} \label{eq:tdl}
\left\|U^{(N)\dag}(s)A_jU^{(N)}(s)-U^{(N+1)\dag}(s)A_jU^{(N+1)}(s)\right\|=e^{-\tilde\Omega(N-j)}.
\end{equation}
Indeed, the difference between $D^{(N)}_k(s')$ and $D^{(N+1)}_k(s')$ for $k\le(j+N)/2$ is controlled by Eq. (\ref{eq:error1}). For $k>(j+N)/2$, the effects of $D^{(N)}_k(s')$ and $D^{(N+1)}_k(s')$ are negligible because the dynamics generated by $D^{(N)}(s')$ satisfies a Lieb-Robinson bound.

Equation (\ref{eq:l11}) follows from Eq. (\ref{eq:tdl}) and the fact that $\psi^{(N)}(0)$ is the reduced state of $\psi^{(N+1)}(0)$ on the first $N$ spins. Equation (\ref{eq:l12}) is obtained by using Eq. (\ref{eq:l11}) twice: In a chain of $N$ spins, we add an $(N+1)$th spin and delete the first spin, introducing errors of $e^{-\tilde\Omega(N-j)}$ and $e^{-\tilde\Omega(j)}$, respectively. The ``therefore'' and ``furthermore'' parts of Lemma \ref{thm:1} are straightforward consequences of Eqs. (\ref{eq:l11}), (\ref{eq:l12}).
\end{proof}

\begin{corollary} \label{cor:op}
Suppose that $H=\{H^{(N)}\}$ is in the trivial phase. The ground-state energy density converges as
\begin{equation}
\left|\frac{\langle H^{(N)}\rangle_N}{N-1}-\lim_{N'\to+\infty}\frac{\langle H^{(N')}\rangle_{N'}}{N'-1}\right|=O(1/N).
\end{equation}
\end{corollary}

\begin{proof}
It follows from
\begin{align}
&\left|\langle H^{(N)}\rangle_N/(N-1)-\langle H^{(N+1)}\rangle_{N+1}/N\right|\nonumber\\
&=\left|\sum_{j=1}^{[N/2]-1}\frac{\langle h_j\rangle_N}{N-1}+\sum_{j=[N/2]}^{N-1}\frac{\langle h_j\rangle_N}{N-1}-\sum_{j=1}^{[N/2]-1}\frac{\langle h_j\rangle_{N+1}}{N}-\frac{\langle h_{[N/2]}\rangle_{N+1}}{N}-\sum_{j=[N/2]+1}^{N}\frac{\langle h_j\rangle_{N+1}}{N}\right|\nonumber\\
&\le\sum_{j=1}^{[N/2]-1}\frac{|\langle h_j\rangle_N-\langle h_j\rangle_{N+1}|}{N-1}+\sum_{j=[N/2]}^{N-1}\frac{|\langle h_j\rangle_N-\langle h_{j+1}\rangle_{N+1}|}{N-1}\nonumber\\
&+\sum_{j=1}^{[N/2]-1}\frac{|\langle h_j\rangle_{N+1}-\langle h_{[N/2]}\rangle_{N+1}|}{(N-1)N}+\sum_{j=[N/2]+1}^{N}\frac{|\langle h_j\rangle_{N+1}-\langle h_{[N/2]}\rangle_{N+1}|}{(N-1)N}\nonumber\\
&\le e^{-\tilde\Omega(N)}+e^{-\tilde\Omega(N)}+\sum_{j=1}^{[N/2]-1}\frac{je^{-\tilde\Omega(-j)}}{(N-1)N}+\sum_{j=[N/2]+1}^{N}\frac{(N+1-j)e^{-\tilde\Omega(-(N+1-j))}}{(N-1)N}=O(1/N^2).
\end{align}
\end{proof}

\begin{theorem} \label{thm:2}
Suppose that $H=\{H^{(N)}\}$ is in the trivial phase. Then, $\langle A_j\rangle_N$ for any $j,N$ and the limits $\lim_{N\to+\infty}\langle A_j\rangle_N$, $\lim_{N\to+\infty}\langle A_{[\alpha N]}\rangle_N$ for any $0<\alpha<1$ can be computed to additive accuracy $\delta$ in time $\poly(1/\delta)$, where the degree of the polynomial is an absolute constant independent of the energy gap.
\end{theorem}

\begin{proof}
Let $B(j,r)$ be the radius-$r$ neighborhood of the $j$th spin. Lemma \ref{thm:1} implies that the expectation value of $A_j$ in the ground state of
\begin{equation}
H^{(N)}|_{B(j,r)}:=\sum_{j\in B(j,r)}h_j
\end{equation}
is an approximation to $\langle A_j\rangle_N$ with error $e^{-\tilde\Omega(r)}$. (Similarly, $\lim_{N\to+\infty}\langle A_j\rangle_N$ can be computed from $H^{(N\to+\infty)}|_{B(j,r)}$, which is supported on a subsystem of finite size.) To achieve additive accuracy $\delta$, it suffices to choose $r=\tilde O(\log(1/\delta))$. Exactly diagonalizing $H^{(N)}|_{B(j,r)}$ would result in an algorithm with running time $e^{O(r)}=e^{\tilde O(\log(1/\delta))}$, which is already close to but still worse than $\poly(1/\delta)$. Since $H^{(N)}|_{B(j,r)}$ has a constant energy gap, we can use the algorithm in Ref. \cite{Hua14a}. The running time of the algorithm is polynomial in the (sub)system size and inverse precision: $\poly(r,1/\delta)=\poly(1/\delta)$, where the degree of the polynomial is an absolute constant independent of the energy gap.
\end{proof}

\begin{remark}
It may be interesting to compare Theorem \ref{thm:2} with a result of Ref. \cite{Hua15}. Only assuming a constant energy gap (not translational invariance or being in the trivial phase), this reference gives an algorithm that computes the ground-state energy density $\langle H^{(N)}\rangle_N/(N-1)$ to additive accuracy $\delta$ in time $\poly(1/\delta)$.
\end{remark}

\subsection{Periodic boundary conditions}

In this subsection only, we consider periodic boundary conditions. For Definition \ref{def:tri} (of the trivial phase), the only modification is that the translation-invariant gapped local Hamiltonians $H(s)=\{H^{(N)}(s)\}$ in the smooth path should also use periodic boundary conditions. Equation (\ref{eq:l12}) becomes trivial: The left-hand side is identically $0$ due to the translational invariance of the ground state. The proof of the following lemma is essentially the same as that of Lemma \ref{thm:1}.

\begin{lemma}
Suppose that $H=\{H^{(N)}\}$ is in the trivial phase. Then,
\begin{equation}
|\langle A_j\rangle_N-\langle A_j\rangle_{N+1}|=e^{-\tilde\Omega(N)}.
\end{equation}
Therefore, $\lim_{N\to+\infty}\langle A_j\rangle_N$ is well defined.
\end{lemma}

\begin{corollary} \label{cor}
Suppose that $H=\{H^{(N)}\}$ is in the trivial phase. Then,
\begin{equation}
\left|\langle h_j\rangle_N-\lim_{N'\to+\infty}\langle h_j\rangle_{N'}\right|=e^{-\tilde\Omega(N)}.
\end{equation}
\end{corollary}

This corollary provides a rigorous justification of the empirical observation that in many one-dimensional translation-invariant gapped systems with periodic boundary conditions, the ground-state energy density converges (almost) exponentially in the system size $N$. In contrast, with open boundary conditions Corollary \ref{cor:op} shows that the scaling is only $O(1/N)$ due to boundary effects.

Theorem \ref{thm:2} remains valid for periodic boundary conditions without any modification.

\section{Higher dimensions} \label{sec:h}

It is very straightforward to extend the results to two and higher spatial dimensions. For notational simplicity and without loss of generality, we consider a two-dimensional square lattice of size $N_x\times N_y$. The thermodynamic limit is defined as a sequence of lattices with growing size $N_x,N_y\to+\infty$. There is one spin at each lattice site, labeled by $(j_x,j_y)$ with $1\le j_x\le N_x$ and $1\le j_y\le N_y$. We use open boundary conditions and define $H=\{H^{(N_x,N_y)}\}$ as a sequence of translation-invariant local Hamiltonians: one for each system size $N_x\times N_y$. Let $\T_x,\T_y$ be the lattice translation operators in the $x,y$ directions, respectively. Let $A_{1,1}$ with $\|A_{1,1}\|\le1$ be a local operator supported in a small neighborhood of site $(1,1)$, and
\begin{equation}
A_{j_x,j_y}:=\mathbb T_y^{j_y-1}\mathbb T_x^{j_x-1}A_{1,1}\mathbb T_x^{-(j_x-1)}\mathbb T_y^{-(j_y-1)}.
\end{equation}
be a local operator supported in a small neighborhood of site $(j_x,j_y)$. Let $\langle\hat O\rangle_{N_x,N_y}$ be the expectation value of an operator $\hat O$ in the ground state of $H^{(N_x,N_y)}$.

Since quasi-adiabatic continuation works in any dimension, Lemma \ref{thm:1} directly generalizes to

\begin{lemma} [open boundary conditions] \label{lm:hd}
Suppose that $H=\{H^{(N_x,N_y)}\}$ is in the trivial phase. Then,
\begin{align}
&|\langle A_{j_x,j_y}\rangle_{N_x,N_y}-\langle A_{j_x,j_y}\rangle_{N_x+1,N_y}|=e^{-\tilde\Omega(N_x-j_x)},\\
&|\langle A_{j_x,j_y}\rangle_{N_x,N_y}-\langle A_{j_x,j_y}\rangle_{N_x,N_y+1}|=e^{-\tilde\Omega(N_y-j_y)},\\
&|\langle A_{j_x,j_y}\rangle_{N_x,N_y}-\langle A_{j_x+1,j_y}\rangle_{N_x,N_y}|=e^{-\tilde\Omega(\min\{j_x,N_x-j_x\})},\\
&|\langle A_{j_x,j_y}\rangle_{N_x,N_y}-\langle A_{j_x,j_y+1}\rangle_{N_x,N_y}|=e^{-\tilde\Omega(\min\{j_y,N_y-j_y\})}.
\end{align}
Therefore, both $\lim_{N_x,N_y\to+\infty}\langle A_{j_x,j_y}\rangle_{N_x,N_y}$ and $\lim_{N_x,N_y\to+\infty}\langle A_{[\alpha_xN_x],[\alpha_yN_y]}\rangle_{N_x,N_y}$ are well defined (and do not depend on the order of limits), where $0<\alpha_x,\alpha_y<1$ are constants. Furthermore, the value of the latter limit is independent of $\alpha_x,\alpha_y$.
\end{lemma}

\begin{theorem}
Suppose that $H=\{H^{(N_x,N_y)}\}$ is in the trivial phase. Then, $\langle A_{j_x,j_y}\rangle_{N_x,N_y}$ for any $j_x,j_y,N_x,N_y$ and the limits $\lim_{N_x,N_y\to+\infty}\langle A_{j_x,j_y}\rangle_{N_x,N_y}$, $\lim_{N_x,N_y\to+\infty}\langle A_{[\alpha_xN_x],[\alpha_yN_y]}\rangle_{N_x,N_y}$ for any $0<\alpha_x,\alpha_y<1$ can be computed to additive accuracy $\delta$ in time $e^{\tilde O(\log^2(1/\delta))}$.
\end{theorem}

\begin{proof}
We exactly diagonalize the Hamiltonian restricted to a $r\times r$ neighborhood of site $(j_x,j_y)$. Lemma \ref{lm:hd} implies that $r=\tilde O(\log(1/\delta))$ suffices, and the running time is $e^{O(r^2)}=e^{\tilde O(\log^2(1/\delta))}$.
\end{proof}

\section*{Acknowledgments}

This work was supported by NSF PHY-1818914.

\appendix

\bibliographystyle{abbrv}
\bibliography{trivial}

\end{document}